\documentclass[12pt]{article}

\usepackage{amssymb}
\usepackage[T1]{fontenc}
\usepackage[utf8]{inputenc}
\usepackage[intlimits]{amsmath}
\usepackage{enumerate}
\usepackage[all]{xy}

\usepackage{amsthm}

\newtheorem{lemma}{Lemma}
\newtheorem{theorem}{Theorem}
\newtheorem{corollary}{Corollary}

\newtheorem{example}{Example}
\usepackage{multirow}

\begin{document}

\begin{center}\Large\bf
%\huge
Deformation of algebroid bracket of differential forms and Poisson manifold
\end{center}
\begin{center}\large
Alina Dobrogowska, Grzegorz Jakimowicz, Karolina Wojciechowicz
\end{center}

\begin{center}
Institute of Mathematics, University of Białystok, Ciołkowskiego 1M, 15-245 Białystok, Poland
\end{center}
\begin{center}
E-mail: alina.dobrogowska@uwb.edu.pl, g.jakimowicz@uwb.edu.pl, kzukowska@math.uwb.edu.pl
\end{center}

\begin{abstract}
We construct the family of algebroid brackets $[\cdot,\cdot]_{c,v}$ on the tangent bundle $T^*M$ to a Poisson manifold $(M,\pi)$ starting from an algebroid bracket of differential forms. We use these brackets to generate Poisson structures on the tangent bundle $TM$.
Next, in the case when $M$ is equipped with a bi-Hamiltonian structure $(M,\pi_1, \pi_2)$ we show how to construct another family of Poisson structures.
Moreover we present how to find Casimir functions for those structures  and we discuss some particular examples.
\end{abstract}

{\bf Keywords:}
Lie algebroid, linear Poisson structure,  bi-Hamiltonian structure, Lie algebra, Casimir function, tangent lift of Poisson structure, Lagrange top\\

\section{Introduction}

The theory of Lie algebroids, see \cite{Ka, 5,13,2,4, We}, is one of important tools of the theory of integrable systems.
The theories of Poisson and bi-Hamiltonian manifolds are another useful tool (see e.g.  \cite{a39, a40, a35,a31,14,10}).
There are links between Poisson manifolds and Lie algebroids.
It is well known that the total space of the dual bundle of a Lie algebroid has a canonical Poisson structure
and there exists the canonical algebroid bracket  of differential forms $A=T^*M$, where $M$ is  a Poisson manifold.
In this paper we consider some modifications of this bracket.
The aim of this paper is to bring together two areas: the theory of Lie algebroids and the theory of Poisson manifolds.
We study the connection between Poisson manifolds and bi-Hamiltonian manifolds,  and some deformation of  Poisson structures generated by  Lie algebroid structures of differential forms.

The paper is organized as follows. In the beginning of
Section 2 we recall the definitions and well known results about Lie algebroids and bi-Hamiltonian manifolds.
Next sections contain the main results of the paper.  In Section 3 we define  the family of algebroid brackets $[\cdot,\cdot]_{c,v}$ (Theorem \ref{theorem-1}) on the tangent bundle $T^*M$ to a Poisson manifold $(M,\pi)$ starting from a algebroid bracket of differential forms. In Section 4 we describe the Poisson structure on the dual bundle of Lie algebroid and we define an additional Poisson structures on $TM$. We show how the bi-Hamiltonian structure from $M$ transfers to the tangent space $TM$.  Moreover we discuss how to lift Casimirs functions and a family of functions in involution  from $M$ to $TM$.
Section 5 contains additional procedures for building Poisson structures on $TM$. 
We also present some examples in Section 6, e.g. the Lagrange top.

\section{Preliminaries and notations}

In the present section  we recall some basic facts about Poisson manifolds, linear Poisson structures, Lie algebroids,  bi-Hamiltonian manifolds and tangent lifts of Poisson structures and bi-Hamiltonian structures.
 
Let $(M,\pi)$ be a $N$--dimensional Poisson manifold. Then the Poisson tensor $\pi\in \Gamma \left(\bigwedge^2 TM\right) $ can be written as 
\begin{equation}
\pi({\bf x})=\sum_{ i,j=1}^{N} \dfrac{1}{2} \pi_{ij}({\bf x})\dfrac{\partial}{\partial x_i} \wedge \dfrac{\partial}{\partial x_j},
\end{equation}
where  ${\bf x}=(x_1,\dots, x_N)$ is a system of local coordinates on $M$.
It leads to the Poisson bracket in the form
\begin{equation}
\{f,g\}({\bf x})=\sum_{ i,j=1}^{N}\pi_{ij}({\bf x})\dfrac{\partial f}{\partial x_i} \dfrac{\partial g}{\partial x_j},
\end{equation}
which is a skew-symmetric bilinear mapping satisfying the Jacobi identity
\begin{equation}
\{\{f,g\},h\}+\{\{h,f\},g\}+\{\{g,h\},f\}=0,
\end{equation}
as well as the Leibniz rule
\begin{equation}
\{ fg,h\}=f\{ g,h \} +g\{ f,h\}.
\end{equation}
The bivector $\pi_{ij}({\bf x})=-\pi_{ji}({\bf x})=\{x_i,x_j\}$ satisfies the following system of equations equivalent to the Jacobi identity
\begin{equation} 
\sum_{s=1}^N\left(
\frac{\partial \pi_{ij}}{\partial x_s}\pi_{sk}+
\frac{\partial \pi_{ki}}{\partial x_s}\pi_{sj}+
\frac{\partial \pi_{jk}}{\partial x_s}\pi_{si}\right)=0.
\end{equation}

Given a Lie algebra $(\mathfrak{g},[\cdot, \cdot])$ of dimension $N$, there exists the canonical Poisson structure on $M=\mathfrak{g}^*$. This bracket is called Lie--Poisson bracket and can be defined by the following formula
\begin{equation}
\{f,g\}({\bf x})=\langle {\bf x},[df({\bf x}), dg({\bf x})]\rangle,
\end{equation}
where $df({\bf x}), dg({\bf x})\in \left(\mathfrak{g}^*\right)^*\simeq  \mathfrak{g}$. There is a natural equivalence between $N$--dimensional linear Poisson structures and $N$--dimensional Lie algebras. If $[e_i,e_j]=\sum_{n=1}^Nc_{ij}^n e_n$ then $\{x_i,x_j\}=\sum_{n=1}^Nc_{ij}^n x_n$,  where $(e_1,e_2,\dots, e_N)$ is a basis of $\mathfrak{g}$ and $c_{ij}^n$ are the structure constants of this Lie algebra.

Let us consider a certain Lie algebroid structure on $T^*M$
\begin{equation}
 \label{d1111xxx}
\xymatrix{
& & A=T^*M \ar@<.0ex>[dd]^*+<1ex>\txt{\tiny{{$q^*_{_{M}}$}}} \ar@<-.0ex>[rr]^*+<1ex>\txt{\tiny{{$a$}}}  & & TM \ar@<.0ex>[dd]^*+<1ex>\txt{\tiny{{$q_{_{M}}$}}}&&\\
    & &&&   & &\\
  && M  \ar@<-.0ex>[rr]^*+<1ex>\txt{\tiny{{$id$}}}  & & M &&\\
 }
 \\
 \end{equation}
A vector bundle map $a$ is called the anchor of the Lie algebroid $A=T^*M$ and in this case it is defined as
\begin{equation}
\label{anchor}
a(df)(\cdot)=\{f,\cdot\}.
\end{equation}
Sections $\Gamma A $ form a Lie algebra with a Lie bracket 
\begin{equation}
\label{bra-alg}
[df,dg]=d\{f,g\},
\end{equation}
 where $f,g\in C^{\infty}(M)$. This bracket must satisfy the following conditions
\begin{align}
& \label{an-1} [df, h\,dg]=h[df,dg]+a(df)(h)dg,\\
& \label{an-2}a\left([df,dg]\right)=[a(df),a(dg)],
\end{align}
for all $df,dg\in \Gamma A$, $h\in C^{\infty}(M)$, see \cite{3,4}. 
%The manifold $M$ is called the base of a Lie algebroid $A$.
On the dual space  $(TM,q_M, M)$ to the Lie algebroid $(T^*M,q^*_M, M)$ we have the  tangent Poisson structure.
The Poisson bracket on $C^{\infty}(TM)$ is given by relations
\begin{align}
& \label{1}\{f\circ q_M, g\circ q_M\}_{TM}=0,\\
& \{ l_{df}, l_{dg}\}_{TM}=l_{[df,dg]},\\
& \label{2}\{f\circ q_M, l_{dg}\}_{TM}=-a(dg)(f)\circ q_M,
\end{align}
where $f,g\in C^{\infty}(M)$.
 In the above formulas $l_{df}\in C^{\infty}(TM)$ is defined by pairing
\begin{equation}
\label{new-var}
l_{df}(\xi)=\big\langle \xi, df(q_M(\xi))\big\rangle, \quad  \xi\in TM.
\end{equation}
In this situation the tangent Poisson tensor can be expressed by formula
\begin{equation}
\label{tensor-pi}
\pi_{TM}({\bf x}, {\bf y})=\left(
\begin{array}{c|c}
0& \pi({\bf x})\\
\hline
\pi({\bf x}) & \sum_{s=1}^{N}\dfrac{\partial\pi }{\partial x_s}({\bf x})y_s
\end{array}
\right),
\end{equation}
where $({\bf x}, {\bf y})=(x_1,\dots, x_N, y_1=l_{dx_1},\dots, y_N=l_{dx_N})$ is a system of local coordinates on $TM$.

Some of the properties of such Poisson structure are well known, see \cite{AJ, GraUrb}.
If $c_1,\dots, c_r$ are  Casimir functions for the Poisson structure $\pi$, i.e. $\{c, f\}=0$ for all $f\in C^{\infty}(M)$,  then the functions 
\begin{equation}
\label{cas}
c_i\circ q_M \quad \textrm{and}\quad l_{dc_i}=\sum_{s=1}^{N}\dfrac{\partial c_i}{\partial x_s}y_s, \quad i=1,\dots r,
\end{equation}
are Casimir functions for the Poisson tensor $\pi_{TM}$. 
Subsequently if the functions $\{H_i\}_{i=1}^k$ are in involution with respect to the Poisson bracket generated by $\pi$,
then the functions 
\begin{equation}
\label{cas-1-n}
H_i\circ q_M^* \quad \textrm{and}\quad  l_{dH_i}=\sum_{s=1}^{N}\dfrac{\partial H_i}{\partial x_s}({\bf x})y_s,\quad i=1,\dots k,
\end{equation}
 are in involution with respect to the Poisson tensor $\pi_{TM}$ given by (\ref{tensor-pi}).
 
We say that  two Poisson tensors $\pi_1$ and $\pi_2$ are compatible if any linear combination
\begin{equation}
\label{bi}
\pi_{\alpha,\beta}=\alpha \pi_1+\beta \pi_2, \quad \alpha,\beta\in\mathbb{R}
\end{equation}
is also a Poisson tensor.
The Poisson structures $\pi_1$ and $\pi_2$ on $M$ are compatible if and only if their Schouten--Nijenhuis bracket vanishes $[\pi_1,\pi_2]_{S-N}=0$, which means that
\begin{equation}
\label{bi-H}
\sum_{s=1}^N\left(\!\!
\pi_{2,sk} \frac{\partial \pi_{1,ij}}{\partial x_s}+
\pi_{1,sk} \frac{\partial \pi_{2,ij}}{\partial x_s}+
\pi_{2,sj} \frac{\partial \pi_{1,ki}}{\partial x_s}+
\right.
\end{equation}
$$\left.
\pi_{1,sj} \frac{\partial \pi_{2,ki}}{\partial x_s}+
\pi_{2,si} \frac{\partial \pi_{1,jk}}{\partial x_s}+
\pi_{1,si} \frac{\partial \pi_{2,jk}}{\partial x_s}\right)=0,
$$
see \cite{8,7}. The manifold $M$ equipped with two compatible Poisson structures $\pi_1$ and $\pi_2$ is called bi-Hamiltonian manifold and we denote it as $(M,\pi_1,\pi_2)$.

As it was shown in \cite{AJ} the algebroid structure (\ref{tensor-pi}) on $TM$  can be deformed as follows
\begin{equation}
\label{tensor-pi-lambda}
\pi_{TM, \lambda}({\bf x}, {\bf y})=\left(
\begin{array}{c|c}
0& \pi_2({\bf x})\\
\hline
\pi_2({\bf x}) & \sum_{s=1}^{N}\dfrac{\partial\pi_2 }{\partial x_s}({\bf x})y_s+\lambda \pi_1 ({\bf x})
\end{array}
\right)
\end{equation}
using bi-Hamiltonian structure $(M,\pi_1,\pi_2)$.
This structure can also be presented globally
\begin{align}
& \label{1n1}\{f\circ q_M, g\circ q_M\}_{TM}=0, \quad \{ l_{df}, l_{dg}\}_{TM}=l_{[df,dg]}+\lambda \{f,g\}_1\circ q_M,\\
& \label{2n2}\{f\circ q_M, l_{dg}\}_{TM}=-a(dg)(f)\circ q_M,
\end{align}
where $f,g\in C^{\infty}(M)$, $\lambda\in\mathbb{R}$. Some of the properties of the Poisson structure above are known, see \cite{AJ}.
If functions $\{H_i\}_{i=1}^k$ are in involution with respect to  the both Poisson brackets given by $\pi_1$ and $\pi_2$,
then the functions (\ref{cas-1-n})  are in involution with respect to the Poisson tensor (\ref{tensor-pi-lambda}).
Moreover, if $c_1,\dots, c_r$, where $r=\dim M-\operatorname{rank} \pi_2$, are  Casimir functions for the Poisson structure $\pi_2$ and functions $f_i^{\lambda}$, $i=1,\dots, r$, satisfy the conditions  
\begin{equation}
\label{q1-n}
\{f^{\lambda}_i,x_j\}_1=\{x_j,c_i\}_2, \quad \textrm{for}\quad j=1,\cdots, n,
\end{equation}
then the functions 
\begin{equation}
\label{cas-b-n}
c_i\circ q_M^* \quad \textrm{and}\quad  \tilde{c}_i=\sum_{s=1}^{N}\dfrac{\partial c_i}{\partial x_s}({\bf x})y_s+\lambda f^{\lambda}_i({\bf x}),\quad i=1,\dots r,
\end{equation}
 are the Casimir functions for the Poisson tensor $\pi_{TM, \lambda}$ given by (\ref{tensor-pi-lambda}).

\section{Deformation of the algebroid bracket of differential forms}

It is well known \cite{3} that there exists the canonical algebroid bracket (\ref{bra-alg}) of differential forms $A=T^*M$, where $M$ is equipped with a Poisson bracket. In this section we consider a some modification of this bracket. We will show that by adding an extra term  to this bracket we will still get the algebroid bracket. We consider a degenerate situation when 
\begin{equation}
\textrm{rank}\; \pi<\dim M.
\end{equation}
In this case Poisson bracket on $M$ has at least one Casimir function.

\begin{theorem}
\label{theorem-1}
Let $(M,\pi)$ be a  Poisson manifold and  $c$ be a Casimir function for $\pi$. If there exists  a vector field  $v\in \Gamma^{\infty}(TM)$ such that Lie derivative of $\pi$ with respect to $v$ vanishes
\begin{equation}
\mathcal{L}_{v}\pi=0,
\end{equation}
then $\left(T^*M,[\cdot, \cdot]_{c,v}, a_{c,v}\right)$ is a Lie algebroid, where Lie bracket
$[\cdot, \cdot]_{c,v}: \Gamma^{\infty}(T^*M)\times \Gamma^{\infty}(T^*M)\longrightarrow \Gamma^{\infty}(T^*M)$
and  the anchor map $a_{c,v}:\Gamma^{\infty}(T^*M)\longrightarrow \Gamma^{\infty}(TM)$
are given by
\begin{align}
&  [\alpha, \beta]_{c,v}=\mathcal{L}_{\pi(\alpha,\cdot)} \beta-\mathcal{L}_{\pi(\beta,\cdot)} \alpha-d(\pi(\alpha, \beta))+c(\beta(v)\mathcal{L}_{v} \alpha      -\alpha(v)\mathcal{L}_{v}\beta),\\
& a_{c,v}(\alpha)(f)=\pi(\alpha,df)-c \alpha(v) df(v),
\end{align}
for all  $\alpha, \beta\in \Gamma^{\infty}(T^*M)$ and $f\in C^{\infty}(M)$.
\end{theorem}

We have divided the proof into two lemmas.
\begin{lemma}
\label{lemma-1}
Suppose that $(M,\pi)$ is a  Poisson manifold, $v\in \Gamma^{\infty}(TM)$ is a vector field and $c\in C^{\infty}(M)$ is a smooth function.
Then 
$\left(T^*M,[\cdot, \cdot]_{v}, a_{v}\right)$ is a Lie algebroid if a Lie bracket
$[\cdot, \cdot]_{v}: \Gamma^{\infty}(T^*M)\times \Gamma^{\infty}(T^*M)\longrightarrow \Gamma^{\infty}(T^*M)$
and  a anchor $a_{v}:\Gamma^{\infty}(T^*M)\longrightarrow \Gamma^{\infty}(TM)$
are given by
\begin{align}
& \label{an-3}  [\alpha, \beta]_{v}= c(\beta(v)\mathcal{L}_{v} \alpha      -\alpha(v)\mathcal{L}_{v}\beta),\\
& \label{an-4} a_{v}(\alpha)(f)=-c \alpha(v) df(v),
\end{align}
where $\alpha, \beta\in \Gamma^{\infty}(T^*M)$ and $f\in C^{\infty}(M)$.
\end{lemma}

\begin{proof}
\begin{enumerate}
\item Let us first prove the formula (\ref{an-1}), i.e. $[\alpha, f\beta]_{v}=f[\alpha, \beta]_{v}+a_{v}(\alpha)(f)\beta$. Using the Leibniz rule for the Lie derivative we calculate
$$
[\alpha, f\beta]_{v} = c(f\beta(v)\mathcal{L}_{v} \alpha      -\alpha(v)\mathcal{L}_{v}(f\beta))=f [\alpha, \beta]_{v}-c\alpha (v)\beta\mathcal{L}_{v} f=$$
$$=f [\alpha, \beta]_{v}+a_{v}(\alpha)(f)\beta.
$$
\item Furthermore, if we prove that (\ref{an-2}), i.e.
\begin{equation}
a_{v}\left([\alpha,\beta]_{v}\right)=[a_{v}(\alpha),a_{v}(\beta)],
\end{equation}
the assertion follows. Using the relationship $\mathcal{L}_{v}\alpha=d\alpha(v)+d\left(\alpha(v)\right)$ we first compute the left side of the above equality
\begin{equation}
a_{v}\left([\alpha,\beta]_{v}\right)(f)=a_{v}\left(c\left( \beta(v)d(\alpha(v))-\alpha(v)d(\beta(v)) \right)\right)(f)=
\end{equation}
$$
=-c^2\left( \beta(v)d(\alpha(v))(v)-\alpha(v)d(\beta(v))(v) \right)df(v).
$$
The right side of the equality is
\begin{equation}
[a_{v} (\alpha), a_{v} (\beta)](f)=a_{v} (\alpha)\left(  a_{v} (\beta)(f) \right)-a_{v} (\beta)\left(  a_{v} (\alpha)(f) \right)=
\end{equation}
$$
=a_{v} (\alpha)\left(  -c \beta (v) df(v) \right)-a_{v} (\beta)\left(   -c \alpha (v) df(v) \right)=
$$$$
=c\alpha (v) d\left( c \beta (v) df(v) \right)(v)-c\beta (v) d\left( c \alpha (v) df(v) \right)(v)=
$$$$
=-c^2\left(  \beta(v)d(\alpha(v))(v)-\alpha(v)d(\beta(v))(v) \right)df(v).
$$
\item We are now in the position to show the Jacobi identity for this bracket 
(this identity results from 1., 2. and below).
 We have
\begin{equation}
 \circlearrowleft [[df,dg]_{v},dh]_{v}=  \circlearrowleft [c\left( dg(v)d \left( df(v) \right)- df(v)d \left( dg(v) \right)\right), dh]_{v}=
\end{equation}
$$
=\circlearrowleft \left(
c[ dg(v)d \left( df(v) \right)- df(v)d \left( dg(v) \right), dh]_{v}-\right.
$$$$\left.
-a_{v}(dh)(c)\left(dg(v)d \left( df(v) \right)- df(v)d \left( dg(v) \right)\right)
\right)=
$$$$
=\circlearrowleft \left( c dg(v)[ d \left( df(v) \right), dh]_{v}-c a_{v}(dh)(dg(v)) d\left( df(v) \right)- \right.
$$$$
-c df(v)[ d \left( dg(v) \right), dh]_{v}+c a_{v}(dh)(df(v)) d\left( dg(v) \right)-
$$$$
\left. -a_{v}(dh)(c)\left(dg(v)d \left( df(v) \right)- df(v)d \left( dg(v) \right) \right)\right)=0.
$$
Here $\circlearrowleft [[f,g]_{v},h]_{v}$ indicates the sum over circular permutations of $f$, $g$, $h$.
This finishes the proof.
\end{enumerate}
\end{proof}

\begin{lemma}
If $c$ is a Casimir function for Poisson bracket $\{\cdot, \cdot\}$ and if $\mathcal{L}_{v}\pi=0$, then the linear combination of the Lie bracket (\ref{bra-alg}) and  the bracket described in Lemma \ref{lemma-1} 
\begin{equation}
[\alpha, \beta]_{c,v}=
[\alpha, \beta]+ [\alpha, \beta]_{v}
\end{equation}
is again the algebroid bracket with the anchor 
\begin{equation}
a_{c,v}(\alpha)(f)=a(\alpha)(f)+a_{v}(\alpha)(f)
\end{equation}
where $\alpha, \beta\in \Gamma^{\infty}(T^*M)$ and $f\in C^{\infty}(M)$.
\end{lemma}
\begin{proof}
\begin{enumerate}
\item Our proof starts with the observation that
\begin{equation}
[\alpha,f\beta]_{c,v}=[\alpha, f\beta]+ [\alpha, f\beta]_{v}=f[\alpha, \beta]+a(\alpha)(f)\beta+
\end{equation}
$$
+ f[\alpha, \beta]_{v}+ a_{v}(\alpha)(f)\beta=f[\alpha, \beta]_{c,v}+a_{c,v}(\alpha)(f)\beta.
$$
\item Next, we show 
that the following equality holds 
\begin{equation}
a_{c,v}\left([\alpha,\beta]_{c,v}\right)=[a_{c,v}(\alpha),a_{c,v}(\beta)].
\end{equation}
We have divided this proof into two steps.
\begin{enumerate}[a)]
\item For differential of  functions. The left side of the equality is
\begin{equation}
L= a_{c,v}\left([df,dg]_{c,v}\right)(h)=
\end{equation}
$$
=a_{c,v}\left(d\{f,g\}+c\left(dg(v)d(df(v))-df(v)d(dg(v))\right)\right)(h)=
$$$$
=\{\{f,g\},h\}+cdg(v)\{df(v),h\}-cdf(v)\{dg(v),h\} -cd\{f,g\}(v)dh(v)-
$$$$
-c^2 dg(v)d(df(v))(v)dh(v)+c^2 df(v)d(dg(v))(v)dh(v).
$$
The Lie derivative of the bi-vector $\pi$ along $v$ is defined by the formula 
$\mathcal{L}_{v}\left(\pi(df,dg)\right) = \left(\mathcal{L}_{v}\pi\right)(df,dg)+\pi (\mathcal{L}_{v}(df), dg)+\pi (df, \mathcal{L}_{v}(dg))=\{df(v),g\}+\{f,dg(v)\}$, where one term of this expression vanish due to the assumption $\mathcal{L}_{v}\pi=0$, then 
$$
L=\{\{f,g\},h\}+cdg(v)\{df(v),h\}-cdf(v)\{dg(v),h\} -cdh(v)\{df(v),g\}-
$$
$$
-cdh(v)\{f,dg(v)\}-c^2 dg(v)d(df(v))(v)dh(v)+c^2 df(v)d(dg(v))(v)dh(v).
$$
The right side of the equality is
\begin{equation}
R=[a_{c,v}(df),a_{c,v}(dg)](h)=
\end{equation}
$$
=a_{c,v}(df)\left(a_{c,v}(dg)(h)\right)-a_{c,v}(dg)\left(a_{c,v}(df)(h)\right)=
$$$$
=a_{c,v}(df)\left(\{g,h\}-cdg(v)dh(v)\right)-a_{c,v}(dg)\left(\{f,h\}-cdf(v)dh(v)\right)=
$$$$
=\{f,\{g,h\}\}-\{f,cdg(v)dh(v)\}-\{g,\{f,h\}\}+\{g,cdf(v)dh(v)\}-
$$$$
-cdf(v)d\left(\{g,h\}-cdg(v)dh(v)\right)(v)+cdg(v)d\left(\{f,h\}-cdf(v)dh(v)\right)(v)=
$$$$
=\{f,\{g,h\}\}+\{g,\{h,f\}\}-cdg(v)\{f,dh(v)\}-cdh(v)\{f,dg(v)\}+
$$$$
+cdf(v)\{g,dh(v)\}+cdh(v)\{g,df(v)\} -cdf(v)d\{g,h\}(v)+cdg (v)d\{f,h\}(v)+
$$$$
+c^2df(v)dh(v)d(dg(v))(v)-c^2dg(v)dh(v)d(df(v))(v)=
$$$$
=\{f,\{g,h\}\}+\{g,\{h,f\}\}-cdh(v)\{f,dg(v)\}+cdh(v)\{g,df(v)\}-
$$$$
-cdf(v)\{dg(v),h\}+cdg(v)\{df(v),h\}+
$$$$
+c^2df(v)dh(v)d(dg(v))(v)-c^2dg(v)dh(v)d(df(v))(v),
$$
because $d\{g,h\}(v)=\mathcal{L}_{v}\{g,h\}=\{dg(v),h\}+\{g,dh(v)\}$.
\item For differential forms.  The left side of the equality is
\begin{equation}
L= a_{c,v}\left([hdf,pdg]_{c,v}\right)=
\end{equation}
$$
= a_{c,v}\left(p[hdf,dg]_{c,v}+a_{c,v}(hdf)(p)dg\right)=
$$$$
= a_{c,v}\left(ph[df,dg]_{c,v}-pa_{c,v}(dg)(h)df+a_{c,v}(hdf)(p)dg\right)=
$$$$
= ph[a_{c,v}(df),a_{c,v}(dg)]-pa_{c,v}(dg)(h)a_{c,v}(df)+ha_{c,v}(df)(p)a_{c,v}(dg).
$$
The right side of the equality is
\begin{equation}
R=[a_{c,v}(hdf),a_{c,v}(pdg)]=[ha_{c,v}(df),pa_{c,v}(dg)]=
\end{equation}
$$
=ha_{c,v}(df)(p)a_{c,v}(dg)+p[ha_{c,v}(df), a_{c,v}(dg)]=
$$$$
=ha_{c,v}(df)(p)a_{c,v}(dg)-pa_{c,v}(dg)(h)a_{c,v}(df)+ph[a_{c,v}(df), a_{c,v}(dg)].
$$
\end{enumerate}
\item At the end we will show 
the Jacobi identity (from 1.,2. and below)
\begin{equation}
\label{Jacobi}
\circlearrowleft [[df,dg]_{c,v},dh]_{c,v}=   \circlearrowleft \left( [[df,dg],dh]_{v}+[[df,dg]_{v},dh]\right)=
\end{equation}
$$
 \circlearrowleft \left( 
c\left(
dh(v)d\left(d\left(\pi(df,dg)\right) (v) \right)-d\left(\pi(df,dg)\right) (v) d\left(dh(v)\right)
\right)\right.+
$$$$
+c[dg(v)d \left( df(v) \right)- df(v)d \left( dg(v) \right), dh]+
$$$$\left.
-a(dh)(c) \left(dg(v)d \left( df(v) \right)- df(v)d \left( dg(v) \right)\right)
\right).
$$
The last term vanishes, because from assumption we have $a(dh)(c)=\{h,c\}=0$. 
The first term using the  relationship $\mathcal{L}_v f=v(f)=df(v)$ can be rewritten in the form
\begin{equation}
c\left(
dh(v)d\left(\mathcal{L}_{v}\left(\pi(df,dg)\right)  \right)-\mathcal{L}_{v}\left(\pi(df,dg)\right)  d\left(dh(v)\right)
\right).
\end{equation}
The Lie derivative of the bi-vector $\pi$ along $v$ is defined by the formula 
$\mathcal{L}_{v}\left(\pi(df,dg)\right) = \left(\mathcal{L}_{v}\pi\right)(df,dg)+\pi (\mathcal{L}_{v}(df), dg)+\pi (df, \mathcal{L}_{v}(dg))$, then 
\begin{equation}
c\big( dh(v)d\left( \left(\mathcal{L}_{v}\pi\right)(df,dg)+\pi (\mathcal{L}_{v}(df), dg)+\pi (df, \mathcal{L}_{v}(dg))  \right)-
\end{equation}
$$
-\left( \left(\mathcal{L}_{v}\pi\right)(df,dg)+\pi (\mathcal{L}_{v}(df), dg)+\pi (df, \mathcal{L}_{v}(dg)) \right)  d\left(dh(v)\right)
\big).
$$
Two terms of this expression vanish due to the second assumption $\mathcal{L}_{v}\pi=0$
\begin{equation}
c\big( dh(v)d\left( \pi (\mathcal{L}_{v}(df), dg)+\pi (df, \mathcal{L}_{v}(dg))  \right)-
\end{equation}
$$
-\left( \pi (\mathcal{L}_{v}(df), dg)+\pi (df, \mathcal{L}_{v}(dg)) \right)  d\left(dh(v)\right)
\big).
$$
Finally the expression (\ref{Jacobi}) becomes
\begin{equation}
\circlearrowleft [[df,dg]_{c,v},dh]_{c,v}=  \circlearrowleft c\big( dh(v)d\left( \pi (\mathcal{L}_{v}(df), dg)+\pi (df, \mathcal{L}_{v}(dg))  \right)-
\end{equation}
$$
-\left( \pi (\mathcal{L}_{v}(df), dg)+\pi (df, \mathcal{L}_{v}(dg)) \right)  d\left(dh(v)\right)+dg(v) d \{df(v),h\}+
$$  $$
-\{h,dg(v)\} d\left(df(v)\right)
-df(v) d \{dg(v),h\}+\{h,df(v)\} d\left(dg(v)\right)
\big)=0.
$$ 
\end{enumerate}
\end{proof}

{\bf Remark:} The algebroid bracket from the Theorem \ref{theorem-1} can be rewritten  in the general form  using the pairing between $TM$ and $T^*M$
\begin{equation}
\big< [\alpha,\beta]_{c,v}, w\big>=\big<\alpha,[\pi,\big<\beta,w\big>]_{S-N}\big> -\big< \beta, [\pi,\big<\alpha,w\big>]_{S-N}\big>
-
\end{equation}
$$
-[\pi,w]_{S-N}(\alpha,\beta)+c\left(\beta(v)\left(\mathcal{L}_{v} \alpha\right)(w)-\alpha(v)\left(\mathcal{L}_{v} \beta\right)(w)\right),
$$
where $\alpha, \beta\in \Gamma^{\infty}(T^*M)$, $v,w\in \Gamma^{\infty}(TM)$, $f\in C^{\infty}(M)$, $[\cdot,\cdot]_{S-N}$ -- the Schouten--Nijenhuis bracket and $\mathcal{L}_{v}\pi=0$, $\pi(dc,\cdot)=0$.

\section{Deformation of tangent Poisson structures  on $TM$}

In this section, we deform Poisson structure  (\ref{tensor-pi}) and tangent lifts of bi-Hamiltonian structure 
(\ref{tensor-pi-lambda}) on $TM$. 

It is well known that a total space of a dual bundle of a Lie algebroid has a canonical Poisson structure.
Thus we obtain the Poisson bracket on $C^{\infty}(TM)$ 
on the tangent bundle $TM$, which is dual to the Lie algebroid $T^*M$ defined in the Theorem 
\ref{theorem-1}
\begin{align}
&\{f\circ q_M, g\circ q_M\}_{c,v}=0, \quad \{ l_{df}, l_{dg}\}_{c,v}=l_{[df,dg]_{c,v}},\\
&\{f\circ q_M, l_{dg}\}_{c,v}=-a_{c,v}(dg)(f)\circ q_M=-a(dg)(f)\circ q_M+ c dg(v)df(v),
\end{align}
where $f,g\in C^{\infty}(M)$, $c$ is Casimir function for $\pi$ and $\mathcal{L}_{v}\pi=0$. In the above formulas $l_{df}\in C^{\infty}(TM)$ is defined by (\ref{new-var}).
In local coordinates $({\bf x}, {\bf y})$ when $v=\sum_{i=1}^{N}v_i\frac{\partial}{\partial x_i}$ the Poisson tensor is given by formula
\begin{equation}
\label{tensor-pi-1-n0}
\begingroup\makeatletter\def\f@size{5}\check@mathfonts 
\pi_{c,v}  ({\bf x}, {\bf y})\!\!=\!\!\!\!\left(\!\!
\begin{array}{c|c}
 0 & \pi ({\bf x})+ c({\bf x})v({\bf x})v^{\top}({\bf x})\\
\hline
\pi ({\bf x})- c({\bf x})v({\bf x})v^{\top}({\bf x}) & \sum_{s=1}^{N}\left(\dfrac{\partial\pi }{\partial x_s}({\bf x})+
c({\bf x})\left( \dfrac{\partial v }{\partial x_s}({\bf x})v^{\top}({\bf x})-v({\bf x})\left( \dfrac{\partial v }{\partial x_s}({\bf x})\right)^{\top}\right)\right)y_s
\end{array}
\!\!\right)\!\!,
\endgroup
\end{equation}
where $v^{\top}=(v_1,\dots, v_{N})$.
The particular case of above construction, when $v=\frac{\partial}{\partial x_p}$, is described by the following theorem.

\begin{theorem}
\label{theorem-2}
Let $(M,\pi)$ be a Poisson manifold and let ${\bf x}=(x_1,\dots, x_N)$ be a system of local coordinates on $M$. If the Poisson tensor $\pi$ does not depend on the variable $x_p$ for certain $1\leq p\leq N$, and the function $c$ is Casimir function for $\pi$,
then we obtain a new Poisson tensor on  $TM$ 
\begin{equation}
\label{tensor-pi-1-n1}
 \pi_{TM, c({\bf x})}  ({\bf x}, {\bf y})=\left(
\begin{array}{c|c}
 0 & \pi ({\bf x})+\delta_{pp} c({\bf x}) \\
\hline
\pi ({\bf x})-\delta_{pp} c({\bf x}) & \sum_{s=1}^{N}\dfrac{\partial\pi }{\partial x_s}({\bf x})y_s
\end{array}
\right),
\end{equation}
where  $({\bf x}, {\bf y})=(x_1,\dots, x_N, y_1,\dots, y_N)$ is a system of local coordinates on $TM$ and $\delta_{ij}$ is  the Kronecker delta. 
\end{theorem}

\begin{proof}
If $c({\bf x})=0$, then we obtain the classical tangent Poisson structure (\ref{tensor-pi}).
For $c({\bf x})\neq 0$, by a direct calculation we obtain
\begin{align} 
& \circlearrowleft \{\{x_i,x_j\}_{TM,c({\bf x})},x_k\}_{TM,c({\bf x})}= \circlearrowleft \{\{x_i,x_j\}_{TM,c({\bf x})},y_k\}_{TM,c({\bf x})}=0,
\\
& \circlearrowleft \{\{y_i,y_j\}_{TM,c({\bf x})},x_k\}_{TM,c({\bf x})}= 
\sum_{m=1}^{N} \delta_{pk}  \left( \delta_{pi}\pi_{mj}({\bf x})-\delta_{pj}\pi_{mi}({\bf x})\right) \dfrac{\partial c}{\partial x_m}({\bf x})+\nonumber\\
&\quad +c({\bf x})\left( \delta_{pi} \dfrac{\partial \pi_{jk}}{\partial x_p}({\bf x})+\delta_{pj} \dfrac{\partial \pi_{ki}}{\partial x_p}({\bf x})- \delta_{pk} \dfrac{\partial \pi_{ij}}{\partial x_p}({\bf x})\right)=0,\nonumber\\
& \circlearrowleft \{\{y_i,y_j\}_{TM,c({\bf x})},y_k\}_{TM,c({\bf x})}=  
\sum_{s=1}^{N}y_s c({\bf x}) \left(
\delta_{pi} \dfrac{\partial^2 \pi_{jk}}{\partial x_p\partial x_s}({\bf x}) 
+\delta_{pj} \dfrac{\partial^2 \pi_{ki}}{\partial x_p\partial x_s}({\bf x})
+\right. \nonumber\\
& \quad \left. + \delta_{pk} \dfrac{\partial^2 \pi_{ij}}{\partial x_p\partial x_s}({\bf x}) \right),\nonumber
\end{align}
because $c$ is Casimir function for $\pi$ and $\pi$ does not depend on the variable $x_p$. 
Here $\circlearrowleft \{\{f,g\},h\}$ indicates the sum over circular permutations of $f$, $g$, $h$.
This finishes the proof.
\end{proof}

\begin{theorem}
\label{theorem-3}
Let $(M,\pi)$ be a Poisson manifold. If the Poisson tensor $\pi$ does not depend on the variable $x_p$ for certain $1\leq p\leq N$, and the function $c$ is linear Casimir function for $\pi$,
then we obtain a new Poisson tensor on  $TM$ 
\begin{equation}
\label{tensor-pi-1-n2}
 \pi_{TM, c({\bf y})}  ({\bf x}, {\bf y})=\left(
\begin{array}{c|c}
 0 & \pi ({\bf x})+\delta_{pp} c({\bf y})\\
\hline
\pi ({\bf x})-\delta_{pp} c({\bf y}) & \sum_{s=1}^{N}\dfrac{\partial\pi }{\partial x_s}({\bf x})y_s
\end{array}
\right).
\end{equation}
\end{theorem}

\begin{proof}
After a direct calculation we obtain
\begin{align} 
& \circlearrowleft \{\{x_i,x_j\}_{TM,c({\bf y})},x_k\}_{TM,c({\bf y})}= 0, \nonumber\\
& \circlearrowleft \{\{x_i,x_j\}_{TM,c({\bf y})},y_k\}_{TM,c({\bf y})}=-\sum_{m=1}^{N} \delta_{pk}  \left( \delta_{pi}\pi_{mj}({\bf x})-\delta_{pj}\pi_{mi}({\bf x})\right) \dfrac{\partial c}{\partial y_m}({\bf y})=0\nonumber\\
& \circlearrowleft \{\{y_i,y_j\}_{TM,c({\bf y})},x_k\}_{TM,c({\bf y})}= 
 \delta_{pk}  \sum_{s,m=1}^{N} y_s\dfrac{\partial}{\partial x_s} 
\left( \delta_{pi}  \pi_{mj}({\bf x})   \dfrac{\partial c}{\partial y_m}({\bf y})- \right.\\
&\quad     -\left. \delta_{pj}  \pi_{mi}({\bf x})   \dfrac{\partial c}{\partial y_m}({\bf y}) \right)  +c({\bf y})\left( \delta_{pi} \dfrac{\partial \pi_{jk}}{\partial x_p}({\bf x})+\delta_{pj} \dfrac{\partial \pi_{ki}}{\partial x_p}({\bf x})- \delta_{pk} \dfrac{\partial \pi_{ij}}{\partial x_p}({\bf x})\right)=0,\nonumber
\end{align}
\begin{align}
& \circlearrowleft \{\{y_i,y_j\}_{TM,c({\bf y})},y_k\}_{TM,c({\bf y})}=  
\sum_{s=1}^{N}y_s c({\bf y}) \dfrac{\partial }{\partial x_s}
 \left( \delta_{pi} \dfrac{\partial \pi_{jk}}{\partial x_p}({\bf x})+\delta_{pj} \dfrac{\partial \pi_{ki}}{\partial x_p}({\bf x})+\right.\nonumber\\
&\quad \left. \delta_{pk} \dfrac{\partial \pi_{ij}}{\partial x_p}({\bf x})\right)=0,\nonumber
\end{align}
because $c$ is linear Casimir function for $\pi$ and $\pi$ does not depend on the variable $x_p$. 
\end{proof}

\begin{theorem}
\label{theorem-4}
If $c_1,\dots, c_r$, where $r=\dim M-\operatorname{rank} \pi$, are Casimir functions for the Poisson structure $\pi$ and if all $c_i$ do not depend on the variable $x_p$ for certain $1\leq p\leq N$, then the functions 
\begin{equation}
\label{cas-n-1}
c_i\circ q_M \quad \textrm{and}\quad l_{dc_i}=\sum_{s=1}^{N}\dfrac{\partial c_i}{\partial x_s}({\bf x})y_s, \quad i=1,\dots r,
\end{equation}
are Casimir functions for the Poisson tensor $\pi_{TM, c({\bf x})}$ and $\pi_{TM, c({\bf y})}$. 
\end{theorem}

\begin{proof}
Let us take the Poisson tensor $\pi_{TM,c({\bf x})}.$ A direct calculation gives us
\begin{align*}
&\{c_i({\bf x}),x_j\}_{TM,c({\bf x})}=0,\\ 
&\{c_i({\bf x}),y_j\}_{TM,c({\bf x})}=\sum_{s=1}^N\dfrac{\partial c_i}{\partial x_s}({\bf x})\pi_{sj}({\bf x})+\delta_{pj}c({\bf x})\dfrac{\partial c_i}{\partial x_p}({\bf x})=0,\\
&\{l_{dc_i},x_j\}_{TM,c({\bf x})}=\sum_{s=1}^N\dfrac{\partial c_i}{\partial x_s}({\bf x})\pi_{sj}({\bf x})-\delta_{pj}c({\bf x})\dfrac{\partial c_i}{\partial x_p}({\bf x})=0,\\
&\{l_{dc_i},y_j\}_{TM,c({\bf x})}=\sum_{s,m=1}^N y_m\dfrac{\partial}{\partial x_m}\left( \pi_{sj}({\bf x})\dfrac{\partial c_i}{\partial x_s}({\bf x})\right)+\delta_{pj}c({\bf x})\dfrac{\partial ^2 c_i}{\partial x_p\partial x_s}({\bf x})y_s=0,
\end{align*}
because $c_i({\bf x})$ is a Casimir function for $\pi$ and $c_i({\bf x})$ does not depend on the variable $x_p.$ The proof for the structure $\pi_{TM,c({\bf y})}$ is completely analogous.
\end{proof}

\begin{theorem}
\label{theorem-5}
Let $(M,\pi_1, \pi_2)$ be a bi-Hamiltonian manifold  and the Poisson tensors $\pi_1$ and $\pi_2$ do not depend on the variable $x_p$ for certain $1\leq p\leq N$.
\begin{enumerate}
       \item If the function $c$ is Casimir function for $\pi_2$, then the structure  on  $TM$ 
\begin{equation}
\label{tensor-pi-lambda-n3}
\pi_{TM, \lambda, c({\bf x})}({\bf x}, {\bf y})=\left(
\begin{array}{c|c}
0& \pi_2({\bf x})+\delta_{pp}c(x)\\
\hline
\pi_2({\bf x})-\delta_{pp} c(x) & \sum_{s=1}^{N}\dfrac{\partial\pi_2 }{\partial x_s}({\bf x})y_s+\lambda \pi_1 ({\bf x})
\end{array}
\right)
\end{equation}
is a Poisson tensor.
        \item If the function $c$ is linear Casimir function for $\pi_2$, then the structure  on  $TM$ 
\begin{equation}
\label{tensor-pi-lambda-n4}
\pi_{TM, \lambda, c({\bf y})}({\bf x}, {\bf y})=\left(
\begin{array}{c|c}
0& \pi_2({\bf x})+\delta_{pp}c(y)\\
\hline
\pi_2({\bf x})-\delta_{pp} c(y) & \sum_{s=1}^{N}\dfrac{\partial\pi_2 }{\partial x_s}({\bf x})y_s+\lambda \pi_1 ({\bf x})
\end{array}
\right)
\end{equation}
 is a Poisson tensor.   
\end{enumerate}
\end{theorem}

\begin{proof}
By direct calculation.  
\end{proof}

\begin{theorem}
\label{theorem-6}
Let $c_1,\dots, c_r$, where $r=\dim M-\operatorname{rank} \pi_2$, be  Casimir functions for the Poisson structure $\pi_2$ and functions $f_i^{\lambda}$, $i=1,\dots, r$, satisfy the conditions  
\begin{equation}
\label{q1-n5}
\{f^{\lambda}_i,x_j\}_2=\{x_j,c_i\}_1, \quad \textrm{for}\quad j=1,\cdots, n.
\end{equation}
If  the functions $c_i$ and $f_i$ do not depend  on the variable $x_p$ for certain $1\leq p\leq N$, then the functions
\begin{equation}
\label{cas-b-n2}
c_i\circ q_M^* \quad \textrm{and}\quad  \tilde{c}_i({\bf x}, {\bf y})=\sum_{s=1}^{N}\dfrac{\partial c_i}{\partial x_s}({\bf x})y_s+\lambda f^{\lambda}_i({\bf x}),\quad i=1,\dots r,
\end{equation}
 are the Casimir functions for the Poisson tensor $\pi_{TM\lambda, c({\bf x})}$ and $\pi_{TM\lambda, c({\bf y})}$ given by (\ref{tensor-pi-lambda-n3}), (\ref{tensor-pi-lambda-n4}), respectively.
\end{theorem}

\begin{proof}
The proof is analogous to the proof of Theorem \ref{theorem-4}.
\end{proof}

Some results were investigated also in the works \cite{Blasz, Tsi-2} for some similar constructions.

In the case of a linear Poisson  structure, when $M=\mathfrak{g}^*$ is the dual to Lie algebra $\mathfrak{g}$, we have additional Poisson structures on $TM$.

\begin{theorem}
\label{theorem-7}
Let $\pi$ be the Lie-Poisson structure on $\mathfrak{g}^*$, which does not depend on the variable $x_p$.
\begin{enumerate}
      \item If $c$ is the Casimir function for $\pi$ then  the tensor
\begin{equation}
\label{naw_m-1}
\widetilde{\pi}_{T\mathfrak{g}^*, c({\bf x})}({\bf x}, {\bf y})=\left(
\begin{array}{c|c}
 \lambda \pi({\bf x}) & \pi({\bf x})+\delta_{pp}c(x)\\
\hline
\pi({\bf x})-\delta_{pp}c(x) & \pi({\bf y})
\end{array}
\right)
\end{equation}
gives a Poisson structure on $T\mathfrak{g}^*$ for any $\lambda\in \mathbb{R}$.
      \item If $c$ is the linear Casimir function for $\pi$ then  the tensors
\begin{equation}
\label{naw_m-2}
\widetilde{\pi}_{T\mathfrak{g}^*, c({\bf y})}({\bf x}, {\bf y})=\left(
\begin{array}{c|c}
 \lambda \pi({\bf x}) & \pi({\bf x})+\delta_{pp}c(y)\\
\hline
\pi({\bf x})-\delta_{pp}c(y) & \pi({\bf y})
\end{array}
\right),
\end{equation}
\begin{equation}
\label{naw_m-3}
\widetilde{\widetilde{\pi}}_{T\mathfrak{g}^*, c({\bf x})}({\bf x}, {\bf y})=\left(
\begin{array}{c|c}
 \lambda \pi({\bf y}) & \pi({\bf x})+\delta_{pp}c(x)\\
\hline
\pi({\bf x})-\delta_{pp}c(x) & \pi({\bf y})
\end{array}
\right),
\end{equation}
\begin{equation}
\label{naw_m-4}
\widetilde{\widetilde{\pi}}_{T\mathfrak{g}^*, c({\bf y})}({\bf x}, {\bf y})=\left(
\begin{array}{c|c}
 \lambda \pi({\bf y}) & \pi({\bf x})+\delta_{pp}c(y)\\
\hline
\pi({\bf x})-\delta_{pp}c(y) & \pi({\bf y})
\end{array}
\right)
\end{equation}
give Poisson structures on $T\mathfrak{g}^*$ for any $\lambda\in \mathbb{R}$.
\end{enumerate}
\end{theorem}

\begin{proof}
By direct calculation.  
\end{proof} 

\begin{theorem}
\label{theorem-8}
Let $c_1,\dots, c_r$, where $r=\dim M-\operatorname{rank} \pi$, be  Casimir functions for the Poisson structure $\pi$ and let  $c_i$ do not depend on the variable $x_p$ for certain $1\leq p\leq N$.  
\begin{itemize}
\item The functions
\begin{equation}
\label{cas-n-1-nm}
c_i({\bf x}) \quad \textrm{and} \quad \tilde{\tilde{c}}_i=c_i({\bf x}-\lambda {\bf y})-c_i({\bf x}), \quad i=1,\dots r,
\end{equation}
are Casimir functions for the Poisson tensor $\widetilde{\pi}_{T\mathfrak{g}^*, c({\bf x})}$ and $\widetilde{\pi}_{T\mathfrak{g}^*, c({\bf y})}$ given by
(\ref{naw_m-1}) and (\ref{naw_m-2}).
\item The functions
\begin{equation}
\label{cas-n-1-nm-1-p}
\hat{c}_i({\bf x}, {\bf y})= c_i({\bf x}-\sqrt{\lambda} {\bf y})+c_i({\bf x}+\sqrt{\lambda} {\bf y}) 
\end{equation}
and
\begin{equation}
\hat{\hat{c}}_i=c_i({\bf x}-\sqrt{\lambda}  {\bf y})-c_i({\bf x}+\sqrt{\lambda} {\bf y}), 
\end{equation}
$i=1,\dots r$, are Casimir functions for the Poisson tensor $\widetilde{\widetilde{\pi}}_{T\mathfrak{g}^*, c({\bf x})}$ and $\widetilde{\widetilde{\pi}}_{T\mathfrak{g}^*, c({\bf y})}$ given by
(\ref{naw_m-3}) and (\ref{naw_m-4}).
\end{itemize}
\end{theorem} 
 
\begin{proof}
By direct calculation.  
\end{proof}

\begin{theorem}
Let $(\mathfrak{g}^{*},\pi_1, \pi_2)$ be a bi-Hamiltonian manifold and the Poisson tensors $\pi_1$ and $\pi_2$ do not depend on the variable $x_p$ for certain $1\leq p\leq N$.
If the function $c$ is Casimir function for $\pi_2$, then the structure  on  $T\mathfrak{g}$ 
\begin{equation}
\label{tensor-pi-lambda-n3-1}
\widetilde{\pi}_{T\mathfrak{g}, \lambda, c({\bf x})}({\bf x}, {\bf y})=\left(
\begin{array}{c|c}
\epsilon \pi_2 ({\bf x})& \pi_2({\bf x})+\delta_{pp}c(x)\\
\hline
\pi_2({\bf x})-\delta_{pp} c(x) & \sum_{s=1}^{N}\dfrac{\partial\pi_2 }{\partial x_s}({\bf x})y_s+\lambda \pi_1 ({\bf x})-\lambda \epsilon \pi_1 ({\bf y})
\end{array}
\right)
\end{equation}
is Poisson tensor.
\end{theorem}

\begin{proof}
By direct calculation.  
\end{proof}

\section{ Exotic deformation of tangent Poisson structures  on $TM$}

In this section, we present additional  Poisson structures  on $TM$. We are discussing also in some particular cases how to transfer Casimir functions and functions in involution from $M$ to the space $TM$.

\begin{theorem}
\label{theorem-9}
Let $(M,\pi_1)$ and $(M,\pi_2)$ be a Poisson manifolds. If the Poisson tensors $\pi_1$ and $\pi_2$ satisfy the conditions
\begin{align}
\label{war-a1}
& \sum_{s=1}^N \left(\pi_{2,sk}({\bf x})\frac{\partial \pi_{1,ij}}{\partial y_s} ({\bf y})
+\pi_{2,sj}({\bf x})\frac{\partial \pi_{1,ki}}{\partial y_s} ({\bf y})
+\pi_{2,si}({\bf x})\frac{\partial \pi_{1,jk}}{\partial y_s}({\bf y}) \right)=0,\\
& \sum_{s=1}^N 
\left(  \sum_{m=1}^N  y_m\frac{\partial \pi_{2,sk}}{\partial x_m} ({\bf x}) \frac{\partial\pi_{1,ij}}{\partial y_s} ({\bf y})
+\pi_{1,si}({\bf y})\frac{\partial \pi_{2,jk}}{\partial x_s} ({\bf x})
+\pi_{1,sj}({\bf y})\frac{\partial \pi_{2,ki}}{\partial x_s} ({\bf x})\right)= 0,\nonumber\\
& \sum_{s,m=1}^N y_s \pi_{1,mk}({\bf y})  \frac{\partial^2\pi_{2,ij}}{\partial x_m\partial x_s} ({\bf x})=0,\nonumber
\end{align}
then there exists a Poisson structure on the Poisson manifold $TM$ associated with $\pi_1$ and $\pi_2$ of the form
\begin{equation}
\label{tensor-pi-1-b}
\left( \pi_1 \ltimes_1 \pi_2 \right) ({\bf x}, {\bf y})=\left(
\begin{array}{c|c}
 \pi_1 ({\bf y})& \pi_2({\bf x})\\
\hline
\pi_2({\bf x}) & \sum_{s=1}^{N}\dfrac{\partial\pi_2 }{\partial x_s}({\bf x})y_s
\end{array}
\right).
\end{equation}
\end{theorem}

\begin{proof}
By direct calculation we obtain
\begin{align} 
1)   &\{\{x_i,x_j\}_{\ltimes_1},x_k\}_{\ltimes_1}+\{\{x_k,x_i\}_{\ltimes_1},x_j\}_{\ltimes_1}+\{\{x_j,x_k\}_{\ltimes_1},x_i\}_{\ltimes_1}=\\
&\sum_{s=1}^N \left(\pi_{2,sk}({\bf x})\frac{\partial \pi_{1,ij}}{\partial y_s} ({\bf y})
+\pi_{2,sj}({\bf x})\frac{\partial \pi_{1,ki}}{\partial y_s} ({\bf y})
+\pi_{2,si}({\bf x})\frac{\partial \pi_{1,jk}}{\partial y_s} ({\bf y})
\right),\nonumber\\
2) &\{\{x_i,x_j\}_{\ltimes_1},y_k\}_{\ltimes_1}+\{\{y_k,x_i\}_{\ltimes_1},x_j\}_{\ltimes_1}+\{\{x_j,y_k\}_{\ltimes_1},x_i\}_{\ltimes_1}=\\
&\sum_{s=1}^N \left( 
\sum_{m=1}^N \frac{\partial \pi_{2,sk}}{\partial x_m} ({\bf x}) y_m \frac{\partial\pi_{1,ij}}{\partial y_s}({\bf y}) 
+\pi_{1,si}({\bf y})\frac{\partial \pi_{2,jk}}{\partial x_s} ({\bf x})
+\pi_{1,sj}({\bf y})\frac{\partial \pi_{2,ki}}{\partial x_s} ({\bf x})\right), \nonumber\\
3) & \{\{y_i,y_j\}_{\ltimes_1},x_k\}_{\ltimes_1}+\{\{y_j,x_k\}_{\ltimes_1},y_i\}_{\ltimes_1}+\{\{x_k,y_i\}_{\ltimes_1},y_j\}_{\ltimes_1}=\\
 &\sum_{s,m=1}^N \pi_{1,mk}({\bf y}) y_s \frac{\partial^2\pi_{2,ij}}{\partial x_m\partial x_s} ({\bf x}), \nonumber\\
4) & \{\{y_i,y_j\}_{\ltimes_1},y_k\}_{\ltimes_1}+\{\{y_j,y_k\}_{\ltimes_1},y_i\}_{\ltimes_1}+\{\{y_k,y_i\}_{\ltimes_1},y_j\}_{\ltimes_1}=0.
\end{align}
\end{proof}

\begin{theorem}
Let $(M,\pi_1)$ and $(M,\pi_2)$ be Poisson manifolds. If the Poisson tensors $\pi_1$ and $\pi_2$ satisfy the conditions
\begin{align}
\label{war-a}
& \sum_{s=1}^N\left(
\pi_{2,sk}({\bf x}) \frac{\partial \pi_{1,ij}}{\partial x_s}({\bf x})+
\pi_{1,si}({\bf x}) \frac{\partial \pi_{2,jk}}{\partial x_s}({\bf x})+
\pi_{1,sj}({\bf x}) \frac{\partial \pi_{2,ki}}{\partial x_s}({\bf x})\right)=0,\\
& \sum_{s,m=1}^N y_s\pi_{1,mk}({\bf x}) \frac{\partial^2 \pi_{2,ij}}{\partial x_m \partial x_s}({\bf x})=0,\nonumber
\end{align}
then there exists a Poisson structure on the Poisson manifold $TM$ associated with $\pi_1$ and $\pi_2$ of the form
\begin{equation}
\label{tensor-pi-1}
\left( \pi_1 \ltimes_2 \pi_2 \right) ({\bf x}, {\bf y})=\left(
\begin{array}{c|c}
\pi_1 ({\bf x})& \pi_2({\bf x})\\
\hline
\pi_2({\bf x}) & \sum_{s=1}^{N}\dfrac{\partial\pi_2 }{\partial x_s}({\bf x})y_s
\end{array}
\right).
\end{equation}
\end{theorem}

\begin{proof}
By the Jacobi identity we have
\begin{equation}
\{\{x_i,x_j\}_{\ltimes_2},x_k\}_{\ltimes_2}+\{\{x_k,x_i\}_{\ltimes_2},x_j\}_{\ltimes_2}+\{\{x_j,x_k\}_{\ltimes_2},x_i\}_{\ltimes_2}=
\end{equation}
$$
=\{\{y_i,y_j\}_{\ltimes_2},y_k\}_{\ltimes_2}+\{\{y_k,y_i\}_{\ltimes_2},y_j\}_{\ltimes_2}+\{\{y_j,y_k\}_{\ltimes_2},y_i\}_{\ltimes_2}=0,
$$
\begin{equation}
\{\{x_i,x_j\}_{\ltimes_2},y_k\}_{\ltimes_2}+\{\{y_k,x_i\}_{\ltimes_2},x_j\}_{\ltimes_2}+\{\{x_j,y_k\}_{\ltimes_2},x_i\}_{\ltimes_2}=
\end{equation}
$$
\sum_{s=1}^N\left(
\frac{\partial \pi_{1,ij}}{\partial x_s}({\bf x})\pi_{2,sk}({\bf x})+
\frac{\partial \pi_{2,ki}}{\partial x_s}({\bf x})\pi_{1,sj}({\bf x})+\right.
$$$$\left.+
\frac{\partial \pi_{2,jk}}{\partial x_s}({\bf x})\pi_{1,si}({\bf x})\right),
$$
\begin{equation}
\{\{y_i,y_j\}_{\ltimes_2},x_k\}_{\ltimes_2}+\{\{x_k,y_i\}_{\ltimes_2},y_j\}_{\ltimes_2}+\{\{y_j,x_k\}_{\ltimes_2},y_i\}_{\ltimes_2}=
\end{equation}
$$
 \sum_{s,m=1}^N  y_s\pi_{1,mk}\pi_{1,sk}({\bf x}) \frac{\partial^2 \pi_{2,ij}}{\partial x_m \partial x_s}({\bf x}).
$$
We see that if the conditions \eqref{war-a}  are fulfilled then $\pi_1 \ltimes_2 \pi_2$ is a Poisson tensor. 
\end{proof}

Let us observe that if we put $\pi_1=0$ then we reduce (\ref{tensor-pi-1-b}) or (\ref{tensor-pi-1})  to (\ref{tensor-pi}), i.e. we obtain the classical tangent Poisson structure on $TM$.
Note that the last condition in  \eqref{war-a1} or \eqref{war-a} is very restrictive but it is automatically fulfilled for the linear or constant Poisson tensor $\pi_2$. In addition, if $\pi_1= \pi_2$, then the conditions \eqref{war-a1} or \eqref{war-a} are   also realized. In this case if $c_1,\dots, c_r$ are Casimir functions for the Poisson structure $\pi_2$, then the functions $c_i({\bf x}-{\bf y})+ c_i({\bf x}+{\bf y})$, $ c_i({\bf x}-{\bf y})- c_i({\bf x}+{\bf y})$, $ i=1,\dots r$,
are the Casimir functions for the Poisson tensor $\pi_2 \ltimes_1 \pi_2$. Similarly the functions $c_i({\bf x}-{\bf y})+ c_i({\bf x})$, $ c_i({\bf x}-{\bf y})- c_i({\bf x})$, $ i=1,\dots r$,
are the Casimir functions for the Poisson tensor $\pi_2 \ltimes_2 \pi_2$.  There is a natural bijection between  linear Poisson structures $\pi_2$ and  Lie algebras $\mathfrak{g}$. Moreover, there is a relationship between linear Poisson structure $\pi_2 \ltimes_1 \pi_2$  and Cartan decomposition of Lie algebra $T\mathfrak{g}=\mathfrak{g}+V$ satisfying the relations
\begin{equation}
[\mathfrak{g},\mathfrak{g}]\subset  \mathfrak{g},\quad 
[\mathfrak{g}, V]\subset  V, \quad [V,V]\subset \mathfrak{g}.
\end{equation}

\begin{corollary} 
If linear Poisson tensors $\pi_1$ and $\pi_2$ are compatible and conditions (\ref{war-a1}) are satisfied  then 
conditions (\ref{war-a}) are also satisfied.
\end{corollary}

\begin{corollary} 
If linear or constant Poisson tensors $\pi_1$ and $\pi_2$ are compatible and conditions (\ref{war-a1}) are satisfied  then 
\begin{equation}
\label{tensor-pi-1-c}
\left( \pi_2 \ltimes_1 \pi_1 \right) ({\bf x}, {\bf y})=\left(
\begin{array}{c|c}
 \pi_2 ({\bf y})& \pi_1({\bf x})\\
\hline
\pi_1({\bf x}) & \sum_{s=1}^{N}\dfrac{\partial\pi_1 }{\partial x_s}({\bf x})y_s
\end{array}
\right)
\end{equation}
is a Poisson tensor.
\end{corollary} 

\begin{corollary} 
If linear or constant Poisson tensors $\pi_1$ and $\pi_2$ are compatible and conditions (\ref{war-a}) are satisfied  then 
\begin{equation}
\label{tensor-pi-1-d}
\left( \pi_2 \ltimes_2 \pi_1 \right) ({\bf x}, {\bf y})=\left(
\begin{array}{c|c}
 \pi_2 ({\bf x})& \pi_1({\bf x})\\
\hline
\pi_1({\bf x}) & \sum_{s=1}^{N}\dfrac{\partial\pi_1 }{\partial x_s}({\bf x})y_s
\end{array}
\right)
\end{equation}
is a Poisson tensor.
\end{corollary}

Note that the conditions (\ref{war-a1}) or (\ref{war-a}) are met for certain classes of Poisson tensors.
One of these classes is described by the following restriction.

\begin{corollary} \label{corollary 1}
If the Poisson tensor $\pi_2$ is constant and  $(M,\pi_1, \pi_2)$ is a bi-Hamiltonian manifold then the conditions 
(\ref{war-a1})  are satisfied.
\end{corollary}
Let us observe that this class is rich, because it contains as $\pi_1$ the Lie-Poisson structure
\begin{equation}
\{f,g\}_{LP}({\bf x})=\bigg< {\bf x},[df({\bf x}),dg({\bf x})]\bigg>, \quad {\bf x}\in\mathfrak{g}^*,
\end{equation}
on the dual $\mathfrak{g}^*$ of a Lie algebra $\mathfrak{g}$, and as $\pi_2$ the frozen Poisson structure
\begin{equation}
\{f,g\}_{{\bf x}_0}({\bf x})=\bigg< {\bf x}_0,[df({\bf x}),dg({\bf x})]\bigg>, 
\end{equation}
where ${\bf x}_0$ is a fixed element of $\mathfrak{g}^*$. These structures are compatible, i.e. form a pencil of Poisson structures for every freezing point ${\bf x}_0$, see \cite{Mis-Fom}.

In general, the Casimir function $c$ for the structure $\pi_1 \ltimes_1 \pi_2$ has to satisfy the following conditions
\begin{align}
&\label{26} \sum_{s=1}^N\left(  
\pi_{1,is}({\bf y}) \frac{\partial c}{\partial x_s} ({\bf x}, {\bf y})
+ \pi_{2,is}({\bf x}) \frac{\partial c}{\partial y_s} ({\bf x}, {\bf y}) \right)=0,\\
& \label{27} \sum_{s=1}^N \pi_{2,is}({\bf x}) \frac{\partial c}{\partial x_s} ({\bf x}, {\bf y})
+ \sum_{s=1}^N y_s \sum_{m=1}^N \frac{\partial \pi_{2,im}}{\partial x_s}({\bf x}) \frac{\partial c}{\partial y_m}({\bf x}, {\bf y}) =0.
\end{align}
In particular, in the case described by the Corollary \ref{corollary 1}, the following theorems can be proved.

\begin{theorem} \label{theorem 2}
Let $\pi_2=const$, $(M,\pi_1, \pi_2)$ be bi-Hamiltonian and  $c_1,\dots, c_r$, where $r=\dim M- \text{rank } \pi_2$, be Casimir functions for the constant Poisson structure $\pi_2$.
Then the functions
\begin{equation}
c_i({\bf y}),\quad \widetilde{c}_i({\bf x},{\bf y})=c_i({\bf x})+\widetilde{\widetilde{c}}_i({\bf y}),
\end{equation}
where $\widetilde{\widetilde{c}}_i$ satisfies the conditions
\begin{equation}
\sum_{s=1}^N\left(  
\pi_{1,js}({\bf y}) \frac{\partial c_i}{\partial x_s} ({\bf x})
+ \pi_{2,js} \frac{\partial \widetilde{\widetilde{c}}_i}{\partial y_s} ( {\bf y}) \right)=0,
\end{equation}
are the Casimir functions for the Poisson tensor $\left( \pi_1 \ltimes_1 \pi_2 \right)$ given by (\ref{tensor-pi-1-b}).
\end{theorem}

\begin{proof}
Proof is obtained by direct calculation from formulas (\ref{26}) and (\ref{27}). 
\end{proof}

\begin{theorem}\label{theorem 3}
Let $\pi_2=const$, $(M,\pi_1, \pi_2)$ be bi-Hamiltonian and functions $\{H_i\}_{i=1}^k$ be in involution with respect to the both Poisson brackets given by $\pi_1$ and $\pi_2$,
then the functions
\begin{equation}
H_i({\bf y}),\quad \text{ and }\quad \widetilde{H}_i({\bf x},{\bf y})=\sum_{s=1}^{N}\dfrac{\partial H_i}{\partial y_s}({\bf y})x_s
\end{equation}
are in involution with respect to the Poisson tensor $\left( \pi_1 \ltimes_1 \pi_2 \right)$ given by (\ref{tensor-pi-1-b}).
\end{theorem}

\begin{proof}
The functions $H_i$ and $H_j$ are in involution with respect to the Poisson structure given by (\ref{tensor-pi-1-b}) when they satisfy the condition
\begin{equation}
\{H_i({\bf x},{\bf y}),H_j({\bf x},{\bf y})\}_{\ltimes_1}=
\sum_{s,m=1}^N\left(  
\pi_{1,sm}({\bf y}) \frac{\partial H_i}{\partial x_s} ({\bf x},{\bf y}) \frac{\partial H_j}{\partial x_m} ({\bf x},{\bf y})\right.+
\end{equation}
$$\left. \pi_{2,sm} \frac{\partial H_i}{\partial x_s} ({\bf x},{\bf y}) \frac{\partial H_j}{\partial y_m} ({\bf x},{\bf y}) +
\pi_{2,sm} \frac{\partial H_i}{\partial y_s} ({\bf x},{\bf y}) \frac{\partial H_j}{\partial x_m} ({\bf x},{\bf y})
\right)=0.$$
From the above we get $\{H_i({\bf y}),H_j({\bf y})\}_{\ltimes_1}=0$. It is also easy to see that
\begin{align}
&\{H_i({\bf y}),\widetilde{H}_j({\bf x},{\bf y})\}_{\ltimes_1}= \sum_{s,m=1}^N
\pi_{2,sm} \frac{\partial H_i}{\partial y_s} ({\bf y}) \frac{\partial H_j}{\partial y_m} ({\bf y})=
\{H_i,H_j\}_2({\bf y})=0,\\
& \{\widetilde{H}_i({\bf x},{\bf y}),\widetilde{H}_j({\bf x},{\bf y})\}_{\ltimes_1}=
\{H_i,H_j\}_1({\bf y})+\sum_{p=1}^{N}x_p \frac{\partial }{\partial y_p}\bigg( \{H_i,H_j\}_2({\bf y})\bigg)=0
\end{align}
from involution with respect to the Poisson tensors $\pi_1$ and $\pi_2$.
\end{proof}

In low-dimensional cases, there is sometimes another possibility to build a family of functions in involution.
\begin{theorem} \label{theorem 4} 
Let $H_i$ be Casimirs functions for the Poisson tensor $\pi_1$ quadratic homogeneous in ${\bf x}$ or linear homogeneous in ${\bf x}$. Then the family of functions
\begin{equation}
\widehat{H}_i({\bf x},{\bf y})=\sum_{s=1}^{N}\dfrac{\partial H_i}{\partial y_s}({\bf y})x_s,
\end{equation}
and the family $\widehat{\widehat{H}}_j(x)$  defined by the following conditions  
\begin{align}
& \{ {H}_i , \widehat{\widehat{H}}_j\}_2= 0,\\
&\sum_{s,m=1}^N\left(  
\pi_{1,sm}({\bf y}) \frac{\partial \widehat{\widehat{H}}_i}{\partial x_s} ({\bf x}) \frac{\partial \widehat{\widehat{H}}_j}{\partial x_m} 
({\bf x})\right)=0,
\end{align} 
are in involution with respect to the Poisson tensor $\left( \pi_1 \ltimes_1 \pi_2 \right)$ given by (\ref{tensor-pi-1-b}).
\end{theorem}

\begin{proof}
From the previous results and the equality $\dfrac{\partial}{\partial y_m}\left(\sum_{s=1}^{N}\dfrac{\partial H_i}{\partial y_s}({\bf y})x_s\right)= \dfrac{\partial}{\partial y_m}\left( \sum_{s=1}^{N}\dfrac{\partial H_i}{\partial x_s}({\bf x})y_s \right)=
\dfrac{\partial H_i}{\partial x_m}({\bf x})$ one can deduce the above theorem.
\end{proof}

\section{Examples}
 
\begin{example}
Lagrange top. 
Let us consider two Poisson structures
\begin{equation}
\pi_1({\bf x})=\left(\begin{array}{ccc}0&\omega x_3&-x_2\\
 -\omega x_3 & 0& x_1\\
 x_2&-x_1&0\end{array}\right), \qquad \pi_2({\bf x})=\left(\begin{array}{ccc} 0&-1&0\\1&0&0\\0&0&0\end{array}\right),
\end{equation}
where ${\bf x}=(x_1,x_2,x_3),\quad {\bf y}=(y_1,y_2,y_3),\quad \omega=const\quad (\omega\neq 0).$ 
In this case, the Casimir function for the structure $\pi_1$ is
\begin{equation}
F({\bf x})=x_1^2+x_2^2+\omega x_3^2
\end{equation}
and for $\pi_2$ 
\begin{equation}
c_1({\bf x})=x_3.
\end{equation}
Let us take as a Hamiltonian
\begin{equation}
H({\bf x})=x_1^2+x_2^2+ x_3^2.
\end{equation} 
The equations of motion for this Hamiltonian computed for first Poisson  structure $\pi_1$ assume the form 
\begin{equation}
\left\{\begin{array}{l}
\dot{x}_1=2(\omega -1)x_2x_3\\
\dot{x}_2=-2(\omega -1)x_1x_3\\
\dot{x}_3=0
\end{array}
\right. .
\end{equation}

It is easy to see that the conditions  (\ref{war-a1}) from Theorem \ref{theorem-9} are satisfied. Then from (\ref{tensor-pi-1-b}),  the tangent Poisson structure is given by
\begin{equation}
\left (\pi_1\ltimes_1 \pi_2\right)(\bf x,\bf y)=\left(\begin{array}{cccccc}0&\omega y_3&-y_2&0&-1&0\\
-\omega y_3&0&y_1&1&0&0\\
y_2&-y_1&0&0&0&0\\
0&-1&0&0&0&0\\
1&0&0&0&0&0\\
0&0&0&0&0&0\end{array}\right).
\end{equation}
The Casimirs from Theorem \ref{theorem 2} for the structure given above assume following form
\begin{align}
\label{casimir-a}
&c_1( {\bf y})=y_3, & \widetilde{c}_1 ({\bf x}, {\bf y}) = x_3-\frac{1}{2}\left(y_1^2+y_2^2+\omega y_3^2\right).
\end{align}

If the functions $H$ and $F$ are in involution with respect to the Poisson structure $\pi_2$, then they satisfy the assumption 
of Theorem \ref{theorem 3}. In this case we obtain
\begin{align}
& H_1({\bf y})=H({\bf y})=y_1^2+y_2^2+ y_3^2,       & & \widetilde{H}_1({\bf x},{\bf y})=2x_1y_1+2x_2y_2+2x_3y_3,\\
& H_2({\bf y})=F({\bf y})=y_1^2+y_2^2+\omega y_3^2, & & \widetilde{H}_2({\bf x},{\bf y})=2x_1y_1+2x_2y_2+2\omega x_3y_3.
\end{align}
Of course, it is also easy to see that two of the four functions can be expressed by the Casimir functions $c_1$ and $\widetilde{c}_1$. Let us take as a Hamiltonian 
\begin{equation}
h=\alpha H_1+\beta H_2+\gamma \widetilde{H}_1+\delta \widetilde{H}_2.
\end{equation}
 The Hamilton's equations, in this case, are given by
\begin{equation}
\left\{\begin{array}{l}
\dot{x}_1=2\gamma (\omega -1)y_2y_3-2(\alpha + \beta )y_2-2(\gamma+\delta)x_2\\
\dot{x}_2=-2\gamma (\omega -1)y_1y_3+2(\alpha + \beta )y_1+2(\gamma+\delta)x_1\\
\dot{x}_3=0\\
\dot{y}_1=-2(\gamma+\delta)y_2\\
\dot{y}_2=2(\gamma+\delta)y_1\\
\dot{y}_3=0
\end{array}
\right. .
\end{equation}

Moreover, since the Casimir function $F$ for the Poisson structure $\pi_1$ is quadratic homogeneous in ${\bf x}$, we obtain that $F$ and $H$ satisfy also the conditions of Theorem
\ref{theorem 4}. In this case we have only two functions
\begin{align}
& \widehat{H}_2({\bf x},{\bf y})=2x_1y_1+2x_2y_2+2\omega x_3y_3 ,       & & \widehat{\widehat{H}}_1({\bf x})=H({\bf x})=x_1^2+x_2^2+ x_3^2.
\end{align}
These functions and Casimir functions (\ref{casimir-a}) are the four constants of motion for the Lagrange top. The Hamiltonian of this top is given by 
\begin{align}
h({\bf x},{\bf y}) & = \frac 12 \left(\widehat{\widehat{H}}_1({\bf x})+(\omega -1) c_1({\bf y}) \widehat{H}_2({\bf x},{\bf y})\right)\\
&=\dfrac{1}{2} \left(x_1^2+x_2^2+ x_3^2\right)+ (\omega -1)y_3 \left(x_1y_1+x_2y_2+\omega x_3y_3 \right).\nonumber
\end{align}
The Hamilton's equations in this case assume the form
\begin{equation}
\left\{\begin{array}{l}
\dot{x}_1=x_2y_3-x_3y_2\\
\dot{x}_2=x_3y_1-x_1y_3\\
\dot{x}_3=x_1y_2-x_2y_1\\
\dot{y}_1=-x_2-(\omega-1)y_2y_3\\
\dot{y}_2=x_1+(\omega-1)y_1y_3\\
\dot{y}_3=0
\end{array}
\right. ,
\end{equation}
see \cite{Med, MorTon, Tsi}. The above equations describe the motion Lagrange top.
\end{example}

\begin{example}

Let us consider all real Lie algebras of dimension equal to three. A complete list of these algebras is given for example by  Mubarakzyanov \cite{Mub}. There are nine real algebras of dimension three, two of which depend on a parameter. Our list is based on the article \cite{n1}.

\begin{center}
\begin{tabular}{|c|c|c|c|c|c|}
\hline Name & Nonzero commutation relations & Invariants   \\
\hline  $\mathcal{A}_{3,1}$   & $[e_2,e_3]=e_1$                                           & $e_1$  \\ 
\hline  $\mathcal{A}_{3,2}$   & $[e_1,e_3]=e_1,\quad [e_2,e_3]=e_1+e_2$                   & $e_1\exp(-e_2/e_1)$  \\ 
\hline  $\mathcal{A}_{3,3}$   & $[e_1,e_3]=e_1,\quad [e_2,e_3]=e_2$                       & $e_2/e_1$  \\ 
\hline  $\mathcal{A}_{3,4}$   & $[e_1,e_3]=e_1,\quad [e_2,e_3]=-e_2$                      & $e_1e_2$  \\ 
\hline  $\mathcal{A}^a_{3,5}$ & $[e_1,e_3]=e_1,\quad [e_2,e_3]=ae_2\quad0<|a|<1)$         & $e_2e_1^{-a}$  \\ 
\hline  $\mathcal{A}_{3,6}$   & $[e_1,e_3]=-e_2,\quad [e_2,e_3]=e_1$                      & $e_1^2+e_2^2$  \\ 
\hline  $\mathcal{A}^a_{3,7}$ & $[e_1,e_3]=ae_1-e_2,\quad [e_2,e_3]=e_1+ae_2\quad(a>0)$   & $(e_1^2+e_2^2)((e_1+ie_2)/(e_1-ie_2))^{ia}$  \\ 
\hline  $\mathcal{A}_{3,8}$   & $[e_1,e_3]=-2e_2,\quad [e_1,e_2]=e_1,\quad [e_2,e_3]=e_3$ & $2e_2^2+e_1e_3+e_3e_1$  \\ 
\hline  $\mathcal{A}_{3,9}$   & $[e_1,e_2]=e_3,\quad [e_2,e_3]=e_1,\quad [e_3,e_1]=e_2$   & $e_1^2+e_2^2+e_3^2$  \\ 
\hline
\end{tabular}
\end{center}

On the dual space $\mathcal{A}_{3,i}^*$, $i=1,\dots ,9$, of a Lie algebra $\mathcal{A}_{3,i}$ is a Poisson structure defined by Poisson tensor $\pi_{3,i}$. The table below presents the corresponding  tensors for the various Lie algebras.
\begin{center}
\begin{tabular}{|c|c|c|c|c|c|}
\hline Name & Poisson tensors     \\
\hline  $\mathcal{A}^*_{3,1}$   & 
\begingroup\makeatletter\def\f@size{4}\check@mathfonts
$\pi_{3,1}=\left(\begin{array}{ccc}
0 & 0 & 0\\
0 & 0 & x_1\\
0 & -x_1 & 0\\\end{array}\right)$ 
\endgroup                                               \\ 
\hline  $\mathcal{A}^*_{3,2}$   &  
\begingroup\makeatletter\def\f@size{4}\check@mathfonts
$\pi_{3,2}=\left(\begin{array}{ccc}
0 & 0 & x_1\\
0 & 0 & x_1+x_2\\
-x_1 & -x_1-x_2 & 0\\\end{array}\right)$ 
\endgroup                    \\ 
\hline  $\mathcal{A}^*_{3,3}$   &   
\begingroup\makeatletter\def\f@size{4}\check@mathfonts
$\pi_{3,3}=\left(\begin{array}{ccc}
0 & 0 & x_1\\
0 & 0 & x_2\\
-x_1 & -x_2 & 0\\\end{array}\right)$ 
\endgroup                       \\ 
\hline  $\mathcal{A}^*_{3,4}$   &  
\begingroup\makeatletter\def\f@size{4}\check@mathfonts
$\pi_{3,4}=\left(\begin{array}{ccc}
0 & 0 & x_1\\
0 & 0 & -x_2\\
-x_1 & x_2 & 0\\\end{array}\right)$ 
\endgroup                          \\ 
\hline  $\left(\mathcal{A}^a_{3,5}\right)^*$ &   
\begingroup\makeatletter\def\f@size{4}\check@mathfonts
$\pi^a_{3,5}=\left(\begin{array}{ccc}
0 & 0 & x_1\\
0 & 0 & ax_2\\
-x_1 & -ax_2 & 0\\\end{array}\right)$ 
\endgroup            \\ 
\hline  $\mathcal{A}^*_{3,6}$   &  
\begingroup\makeatletter\def\f@size{4}\check@mathfonts
$\pi_{3,6}=\left(\begin{array}{ccc}
0 & 0 & -x_2\\
0 & 0 & x_1\\
x_2 & -x_1 & 0\\\end{array}\right)$ 
\endgroup                       \\ 
\hline  $\left(\mathcal{A}^a_{3,7}\right)^*$ & 
\begingroup\makeatletter\def\f@size{4}\check@mathfonts
$\pi^a_{3,7}=\left(\begin{array}{ccc}
0 & 0 & ax_1-x_2\\
0 & 0 & x_1+ax_2\\
-ax_1+x_2 & -x_1-ax_2 & 0\\\end{array}\right)$ 
\endgroup      \\ 
\hline  $\mathcal{A}^*_{3,8}$   &  
\begingroup\makeatletter\def\f@size{4}\check@mathfonts
$\pi_{3,8}=\left(\begin{array}{ccc}
0 & x_1 & -2x_2\\
-x_1 & 0 & x_3\\
2x_2 & -x_3 & 0\\\end{array}\right)$ 
\endgroup    \\ 
\hline  $\mathcal{A}^*_{3,9}$   & 
\begingroup\makeatletter\def\f@size{4}\check@mathfonts
$\pi_{3,9}=\left(\begin{array}{ccc}
0 & x_3 & -x_2\\
-x_3 & 0 & x_1\\
x_2 & -x_1 & 0\\\end{array}\right)$ 
\endgroup       \\ 
\hline
\end{tabular}
\end{center}

The next table describes if the above structures are compatible in sense (\ref{bi}), i.e. $\mathbb{R}^3$ equipped with these is a bi-Hamiltonian manifold.
\begin{center}
\begingroup\makeatletter\def\f@size{4}\check@mathfonts
\begin{tabular}{ccccccccc|c|}
\hline \multicolumn{1}{ |c|  }{ $\mathcal{A}^*_{3,1}$ } & \multicolumn{1}{ c|  }{$\mathcal{A}^*_{3,2}$} & \multicolumn{1}{ |c|  }{$\mathcal{A}^*_{3,3}$} & \multicolumn{1}{ |c|  }{$\mathcal{A}^*_{3,4}$} & \multicolumn{1}{ |c|  }{$\left(\mathcal{A}^a_{3,5}\right)^*$} & \multicolumn{1}{ |c|  }{ $\mathcal{A}^*_{3,6}$ }  &  \multicolumn{1}{ |c|  }{$\left(\mathcal{A}^a_{3,7}\right)^*$} & \multicolumn{1}{ |c|  }{$\mathcal{A}^*_{3,8}$} & $\mathcal{A}^*_{3,9}$ & Name  \\
\hline   \multicolumn{1}{ |c|  }{YES} & \multicolumn{1}{ |c|  }{YES} &\multicolumn{1}{ |c|  }{YES}&\multicolumn{1}{ |c|  }{YES}&\multicolumn{1}{ |c|  }{YES}&\multicolumn{1}{ |c|  }{YES}&\multicolumn{1}{ |c|  }{YES}&\multicolumn{1}{ |c|  }{YES}&\multicolumn{1}{ |c|  }{YES}& $\mathcal{A}^*_{3,1}$  \\ 
\cline{1-10}    &\multicolumn{1}{ |c|  }{YES}   &\multicolumn{1}{ |c|  }{YES}&\multicolumn{1}{ |c|  }{YES}&\multicolumn{1}{ |c|  }{YES}&\multicolumn{1}{ |c|  }{YES}&\multicolumn{1}{ |c|  }{YES}&\multicolumn{1}{ |c|  }{NO}& \multicolumn{1}{ |c|  }{NO}&  $\mathcal{A}^*_{3,2}$\\ 
\cline{2-10}  & &\multicolumn{1}{ |c|  }{YES}&\multicolumn{1}{ |c|  }{YES}&\multicolumn{1}{ |c|  }{YES}&\multicolumn{1}{ |c|  }{YES}&\multicolumn{1}{ |c|  }{YES}&\multicolumn{1}{ |c|  }{NO}&\multicolumn{1}{ |c|  }{NO}&  $\mathcal{A}^*_{3,3}$ \\ 
\cline{3-10}    &  &&\multicolumn{1}{ |c|  }{YES}&\multicolumn{1}{ |c|  }{YES}&\multicolumn{1}{ |c|  }{YES}&\multicolumn{1}{ |c|  }{YES}&\multicolumn{1}{ |c|  }{YES}&\multicolumn{1}{ |c|  }{YES}&  $\mathcal{A}^*_{3,4}$\\ 
\cline{4-10}  &&&&\multicolumn{1}{ |c|  }{YES}&\multicolumn{1}{ |c|  }{YES} &\multicolumn{1}{ |c|  }{YES}&\multicolumn{1}{ |c|  }{NO}&\multicolumn{1}{ |c|  }{NO}& $\left(\mathcal{A}^a_{3,5}\right)^*$  \\ 
\cline{5-10}   &&&& &\multicolumn{1}{ |c|  }{YES}&\multicolumn{1}{ |c|  }{YES}&\multicolumn{1}{ |c|  }{YES}&\multicolumn{1}{ |c|  }{YES}& $\mathcal{A}^*_{3,6}$  \\ 
\cline{6-10}   &&&  &&&\multicolumn{1}{ |c|  }{YES}&\multicolumn{1}{ |c|  }{NO}&\multicolumn{1}{ |c|  }{NO}& $\left(\mathcal{A}^a_{3,7}\right)^*$\\ 
\cline{7-10}   &&  &&&&&\multicolumn{1}{ |c|  }{YES}&\multicolumn{1}{ |c|  }{YES}&  $\mathcal{A}^*_{3,8}$ \\ 
\cline{8-10}    & &&&&&&&\multicolumn{1}{ |c|  }{YES}& $\mathcal{A}^*_{3,9}$  \\ 
\cline{9-10}
\end{tabular}
\endgroup
\end{center}

We can apply the analogous procedure as in Example 1 to other three dimensional Lie algebras.
We will take the linear Poisson tensor on the manifold $\mathcal{A}^*_{3,8}$ and frozen Poisson tensor compatible with it
\begin{equation}
\pi_{1}({\bf x})=\left(\begin{array}{ccc}
0 & x_1 & -2x_2\\
-x_1 & 0 & x_3\\
2x_2 & -x_3 & 0\\\end{array}\right), \qquad \pi_2({\bf x})=\left(\begin{array}{ccc} 0&1&0\\-1&0&0\\0&0&0\end{array}\right).
\end{equation}
The Casimir function for the structure $\pi_1$ is 
\begin{equation}
H({\bf x})=x_2^2+x_1x_3
\end{equation}
and for $\pi_2$
\begin{equation}
c_1({\bf x})=x_3.
\end{equation}
The Casimirs from Theorem \ref{theorem 2} for the structure 
\begin{equation}
\label{structure}
\left (\pi_1\ltimes_1 \pi_2\right)(\bf x,\bf y)=\left(\begin{array}{cccccc}
0&y_1&-2y_2&0&1&0\\
-y_1&0&y_3&-1&0&0\\
2y_2&-y_3&0&0&0&0\\
0&1&0&0&0&0\\
-1&0&0&0&0&0\\
0&0&0&0&0&0\end{array}\right),
\end{equation}
assume following form
\begin{align}
\label{casimir-a3}
&c_1( {\bf y})=y_3, & \widetilde{c}_1 ({\bf x}, {\bf y}) = x_3+y_2^2+y_1 y_3.
\end{align}
The Casimir function $H$ for the Poisson structure $\pi_1$ is quadratic homogeneous in ${\bf x}$. Moreover, this function and
\begin{equation}
\label{Ha}
\widehat{\widehat{H}}_1({\bf x})=x_2^2+ x_1 x_3+\lambda x_3^2
\end{equation}
are in involution with respect to the Poisson tensor $\pi_2$, where $\lambda$ is any constant. So from Theorem \ref{theorem 4} we obtain that 
(\ref{Ha}) and
\begin{equation}
 \widehat{H}_2({\bf x},{\bf y})=2x_2y_2+ x_1y_3+x_3y_1 
\end{equation}
are in involution with respect to the Poisson structure (\ref{structure}).
If we take as the Hamiltonian  
\begin{equation}
h({\bf x},{\bf y}) = x_2^2+ x_1 x_3+\lambda x_3^2+\alpha y_3\left( 2 x_2 y_2+ x_1 y_3+x_3 y_1 \right)
\end{equation}
we obtain the following Hamilton's equations 
\begin{equation}
\left\{\begin{array}{l}
\dot{x}_1=2\alpha x_2y_3-4\lambda x_3y_2+2y_1x_2-2x_1y_2\\
\dot{x}_2= x_1y_3-x_3y_1+(2\lambda-\alpha ) x_3y_3\\
\dot{x}_3=2(x_3y_2-x_2y_3)\\
\dot{y}_1=2x_2+2\alpha y_2y_3\\
\dot{y}_2=-x_3-\alpha y^3_3\\
\dot{y}_3=0
\end{array}
\right. ,
\end{equation}
where $\alpha$ is an arbitrary constant.

\end{example}

\begin{example}
The table given below shows which Poisson structures related to the Lie algebras of dimension $3$, see Example 2, can be used to build a Poisson structure of dimension $6$. A table row $\mathcal{A}^*_{3,i}$ and table column $\mathcal{A}^*_{3,j}$ means that, we are building a Poisson tensor (Theorem \ref{theorem-9})
\begin{equation}
\left (\pi_{3,i}\ltimes_1 \pi_{3,j}\right)=\left(\begin{array}{c|c}
\pi_{3,i}({\bf y}) & \pi_{3,j}({\bf x})\\
\hline
\pi_{3,j}({\bf x}) & \pi_{3,j}({\bf y})\\\end{array}\right).
\end{equation}

\begin{center}
\begingroup\makeatletter\def\f@size{4}\check@mathfonts
\begin{tabular}{|c|c|c|c|c|c|c|c|c|c|}
\hline Name & $\mathcal{A}^*_{3,1}$  & $\mathcal{A}^*_{3,2}$ & $\mathcal{A}^*_{3,3}$ & $\mathcal{A}^*_{3,4}$ & $\left(\mathcal{A}^a_{3,5}\right)^*$ &  $\mathcal{A}^*_{3,6}$   &  $\left(\mathcal{A}^a_{3,7}\right)^*$ & $\mathcal{A}^*_{3,8}$ & $\mathcal{A}^*_{3,9}$  \\
\hline  $\mathcal{A}^*_{3,1}$   & YES & YES & YES & NO  & NO  &  NO &  NO & NO  &NO  \\ 
\hline  $\mathcal{A}^*_{3,2}$   & YES & YES & YES &  NO & NO  &  NO &  NO & NO  & NO  \\ 
\hline  $\mathcal{A}^*_{3,3}$   & YES & YES & YES & YES &YES  &YES  & YES &  NO & NO\\ 
\hline  $\mathcal{A}^*_{3,4}$   &  NO &  NO & YES & YES & YES &  NO & NO  &  NO & NO\\ 
\hline  $\left(\mathcal{A}^a_{3,5}\right)^*$ &NO& NO&YES & YES& YES & NO  & NO  &   NO  &  NO \\ 
\hline  $\mathcal{A}^*_{3,6}$   &  NO &  NO &  YES&  NO &  NO & YES & YES &  NO & NO \\ 
\hline  $\left(\mathcal{A}^a_{3,7}\right)^*$ &NO& NO& YES & NO & NO  & YES & YES &  NO  & NO  \\ 
\hline  $\mathcal{A}^*_{3,8}$   &  NO &  NO & NO  & NO  &  NO &  NO & NO  & YES &   NO \\ 
\hline  $\mathcal{A}^*_{3,9}$   &  NO & NO  & NO  & NO &  NO & NO  & NO  & NO  & YES \\ 
\hline
\end{tabular}
\endgroup
\end{center}

Let us consider, in table given above, the row $\mathcal{A}^*_{3,3}$ and the column $\mathcal{A}^*_{3,1}.$ It means we are building a Poisson tensor (Theorem \ref{theorem-9})
\begin{equation}
\left (\pi_{3,3}\ltimes_1 \pi_{3,1}\right)({\bf x},{\bf y})=\left(\begin{array}{c|c}
\pi_{3,3}({\bf y}) & \pi_{3,1}({\bf x})\\
\hline
\pi_{3,1}({\bf x}) & \pi_{3,1}({\bf y})\\\end{array}\right).
\end{equation}
In local coordinates ${\bf x}=(x_1,x_2,x_3), {\bf y}=(y_1,y_2,y_3)$ we get
\begin{equation}
\left (\pi_{3,3}\ltimes_1 \pi_{3,1}\right)({\bf x},{\bf y})=\left(\begin{array}{ccc|ccc}
0&0&y_1&0&0&0\\0&0&y_2&0&0&x_1\\-y_1&-y_2&0&0&-x_1&0\\
\hline
0&0&0&0&0&0\\0&0&x_1&0&0&y_1\\0&-x_1&0&0&-y_1&0\end{array}\right).
\end{equation}
We can prove by direct calculation and changing the variables  that this is a tensor for a Lie-Poisson structure related to the Lie algebra $\mathcal{A}_{6,16}$ from the classification given in \cite{n1}. The commutation relations for $\mathcal{A}_{6,16}$
are $[e_1,e_3]=e_4$, $[e_1,e_4]=e_5$, $[e_1,e_5]=e_6$, $[e_2,e_3]=e_5$ and $[e_2,e_4]=e_6$, where 
$(x_1,x_2,x_3,y_1,y_2,y_3)\mapsto (-e_5,-e_3, e_1, e_6,e_4, -e_2)$. In this case, the Casimirs assume the following form $c_1({\bf x})=\frac{x_2}{x_1}$ for $\pi_1$ and  $c_2({\bf x})=x_1$ for $\pi_2$, respectively. Moreover, the Casimirs for the $\pi_1 \ltimes_1 \pi_2$ are given by the formulas $c_1({\bf x}, {\bf y})=y_1$ and $c_2({\bf x}, {\bf y})=\frac 13 x_1^3-x_1y_1y_2+x_2y_1^2$.

Let us take now the row $\mathcal{A}^*_{3,2}$ and the column $\mathcal{A}^*_{3,1}.$ It means we are building a Poisson tensor 
\begin{equation}
\left (\pi_{3,2}\ltimes_1 \pi_{3,1}\right)({\bf x},{\bf y})=\left(\begin{array}{c|c}
\pi_{3,2}({\bf y}) & \pi_{3,1}({\bf x})\\
\hline
\pi_{3,1}({\bf x}) & \pi_{3,1}({\bf y})\\\end{array}\right).
\end{equation}
In local coordinates ${\bf x}=(x_1,x_2,x_3), {\bf y}=(y_1,y_2,y_3)$ we get
\begin{equation}
\left (\pi_{3,2}\ltimes_1 \pi_{3,1}\right)({\bf x},{\bf y})=\left(\begin{array}{ccc|ccc}
0&0&y_1&0&0&0\\0&0&y_1+y_2&0&0&x_1\\-y_1&-y_1-y_2&0&0&-x_1&0\\
\hline
0&0&0&0&0&0\\0&0&x_1&0&0&y_1\\0&-x_1&0&0&-y_1&0\end{array}\right).
\end{equation}
We can prove by direct calculation and changing the variables  that this is also a tensor for a Lie-Poisson structure related also to the Lie algebra $\mathcal{A}_{6,16}$, where 
$(x_1,x_2,x_3,y_1,y_2,y_3)\mapsto (x_1,x_2,x_3,y_1,y_1+y_2,y_3)=(-e_5,-e_3, e_1, e_6,e_4, -e_2)$. In this case, the Casimirs assume the following form $c_1({\bf x})=x_1$ for $\pi_1$ and  $c_2({\bf x})=\frac{x_2}{x_1}$ for $\pi_2$, respectively. Moreover, the Casimirs for the $\pi_1 \ltimes_1 \pi_2$ are given by the formulas $c_1({\bf x}, {\bf y})=y_1$ and $c_2({\bf x}, {\bf y})=\frac 13 x_1^3-x_1y_1(y_1+y_2)+x_2y_1^2$.

We will get very important cases, the Poisson structure on $\mathfrak{so(4)}^*$, if we consider the following case
\begin{equation}
\left (\pi_{3,9}\ltimes_1 \pi_{3,9}\right)({\bf x},{\bf y})=\left(\begin{array}{c|c}
\pi_{3,9}({\bf y}) & \pi_{3,9}({\bf x})\\
\hline
\pi_{3,9}({\bf x}) & \pi_{3,9}({\bf y})\\\end{array}\right)\cong \pi_{3,9}({\bf x}+{\bf y})\oplus\pi_{3,9}({\bf y}-{\bf x}),
\end{equation}
and $\mathfrak{e(3)}^*$ if we take
\begin{equation}
\pi({\bf x},{\bf y})=\left(\begin{array}{c|c}
0 & \pi_{3,9}({\bf x})\\
\hline
\pi_{3,9}({\bf x}) & \pi_{3,9}({\bf y})\\\end{array}\right),
\end{equation}
see \cite{Adler, AJ}.
\end{example}

\begin{example}
If we take Euclidean Lie algebra $\mathcal{A}_{3,6}=\mathfrak{e}(2)$, then using the constructions described in Theorems \ref{theorem-2}, 
\ref{theorem-5} and \ref{theorem-7}  we obtain the following
Poisson structures on $T\mathcal{A}_{3,6}^*$.
\begin{enumerate}
\item From Theorem \ref{theorem-2} we have
\begin{equation}
\pi_{TM,c({\bf x})}({\bf x},{\bf y})\!\!=\!\!\!
\left( \!\!\!\begin{array}{ccc|ccc}
0 & 0 & 0 &    0 &  0  & -x_2\\
0 & 0 & 0 &    0 &  0  &  x_1\\
0 & 0 & 0 &  x_2 &-x_1 &    \epsilon (x_1^2+x_2^2)\\
\hline
0 &  0  & -x_2                         & 0   & 0   & -y_2\\
0 &  0  &  x_1                         & 0   & 0   &  y_1\\
x_2 &-x_1 & \epsilon (x_1^2+x_2^2)     & y_2 &-y_1 &   0
\end{array}\!\!\!\right),
\end{equation}
where $\epsilon$ is an arbitrary constant. 
In this case, the Casimir functions (Theorem \ref{theorem-4}) are given by the formulas
\begin{equation}
c_1({\bf x},{\bf y})= x_1^2+x_2^2, \quad l_{dc_1}=2x_1y_1+2x_2y_2.
\end{equation}
\item From Theorem \ref{theorem-5} we have
\begin{equation}
\pi_{TM,c({\bf x})}({\bf x},{\bf y})\!\!=\!\!\!
\left(\!\!\! \begin{array}{ccc|ccc}
0 & 0 & 0 &    0 &  0  & -x_2\\
0 & 0 & 0 &    0 &  0  &  x_1\\
0 & 0 & 0 &  x_2 &-x_1 &    \epsilon (x_1^2+x_2^2)\\
\hline
0 &  0  & -x_2                         & 0   & 0   & -y_2+\lambda x_1\\
0 &  0  &  x_1                         & 0   & 0   &  y_1-\lambda x_2\\
x_2 &-x_1 & \epsilon (x_1^2+x_2^2)     & y_2 &-y_1 &   0
\end{array}\!\!\!\right),
\end{equation}
where we use the compatible Poisson structure related to Lie algebra $\mathcal{A}_{3,4}$.
In this case, the Casimir functions (Theorem \ref{theorem-6}) are given by the formulas
\begin{equation}
c_1({\bf x},{\bf y})= x_1^2+x_2^2, \quad \tilde{c}_1=2x_1y_1+2x_2y_2-2\lambda x_1x_2.
\end{equation}
\item From Theorem \ref{theorem-7} we have
\begin{equation}
\pi_{TM,c({\bf x})}({\bf x},{\bf y})\!\!=\!\!\!
\left( \!\!\!\begin{array}{ccc|ccc}
0 & 0 & -\lambda x_2 &    0 &  0  & -x_2\\
0 & 0 & \lambda x_1 &    0 &  0  &  x_1\\
\lambda x_2 & \lambda x_1 & 0 &  x_2 &-x_1 &    \epsilon (x_1^2+x_2^2)\\
\hline
0 &  0  & -x_2                         & 0   & 0   & -y_2\\
0 &  0  &  x_1                         & 0   & 0   &  y_1\\
x_2 &-x_1 & \epsilon (x_1^2+x_2^2)     & y_2 &-y_1 &   0
\end{array}\!\!\!\right).
\end{equation}
In this case, the Casimir functions (Theorem \ref{theorem-8}) are given by the formulas
\begin{equation}
c_1({\bf x},{\bf y})= x_1^2+x_2^2, \quad \tilde{\tilde{c}}_1=\lambda^2 y_1^2+\lambda^2 y_2^2-2\lambda x_1y_1-2\lambda x_2y_2.
\end{equation}
\end{enumerate}

\end{example}

\section*{Acknowledgments}

This article has received financial support from the Polish Ministry of Science and Higher Education under subsidy for maintaining the research potential of the Faculty of Mathematics and Informatics, University of Bialystok (BST-148).

\bibliographystyle{plain}

\end{document}